\documentclass[11pt,letterpaper]{article}   
\usepackage{graphicx}
\usepackage[singlespacing]{setspace}
\usepackage{amsmath, amsthm, amsfonts, enumerate, amssymb, setspace}
\usepackage{acronym, cite}
\usepackage[margin=1in]{geometry}
\usepackage{algorithm}
\usepackage[noend]{algpseudocode}

\algrenewcommand\algorithmicrequire{\textbf{Input:}}
\algrenewcommand\algorithmicensure{\textbf{Output:}}
\usepackage[dvipsnames]{xcolor}
\usepackage{hyperref}
\usepackage{cleveref}

\newcommand\myshade{85}
\colorlet{mylinkcolor}{RoyalBlue}
\colorlet{mycitecolor}{SeaGreen}
\colorlet{myurlcolor}{Aquamarine}

\hypersetup{
	linkcolor  = mylinkcolor!\myshade!black,
	citecolor  = mycitecolor!\myshade!black,
	urlcolor   = myurlcolor!\myshade!black,
	colorlinks = true,
}

\date{June 26, 2019}
\title{\sc Approximate Dynamic Programming For Linear Systems with State and Input Constraints \\ \bigskip \large [\textcolor{blue}{Full version of paper presented at European Control Conference 2019}]}
\author{Ankush Chakrabarty\thanks{Corresponding author. Email: \texttt{chakrabarty@merl.com}. Phone: +1~(617)~758-6175.}, Rien Quirynen, Claus Danielson, Weinan Gao
	\thanks{A. Chakrabarty, R. Quirynen, and C. Danielson are affiliated with the Control and Dynamical (CD) Systems Group, Mitsubishi Electric Research Laboratories, Cambridge, MA, USA. W. Gao is with the Department of Electrical and Computer Engineering, Allen E. Paulson College of Engineering and Computing, Georgia Southern University, Statesboro, GA 30460, USA.}
}

\usepackage{algorithm}
\usepackage{algpseudocode}
\newcommand\ignore[1]{{}}

\DeclareMathOperator*{\argmin}{arg\,min}

\newtheorem{theorem}{Theorem}

\newtheorem{definition}{Definition}
\newtheorem{assumption}{Assumption}
\theoremstyle{definition}

\newtheorem{remarkx}{Remark}
\newenvironment{remark}
{\pushQED{\qed}\remarkx}
{\popQED\endremarkx}

\usepackage{bbm}

\usepackage{color}
\definecolor{wheat}{rgb}{0.96,0.87,0.70}

\newcommand{\xk}{\bar{x}_{t}}
\newcommand{\Kk}{\bar{K}_{t}}
\newcommand{\uk}{\bar{u}_{t}}
\newcommand{\ukt}{\tilde{u}_{t}}
\newcommand{\xkn}{\bar{x}_{t+1}}
\newcommand{\xknt}{\tilde{x}_{t+1}}
\newcommand{\unoise}{\nu_{t}}
\newcommand{\objfunc}{\mathcal J_t(P)}
\newcommand{\ztp}{\bar{\mathcal J}_t(P)}

\newcommand{\Kopt}{K_\infty}
\newcommand{\Popt}{P_\infty}

\newcommand{\eye}{\mathbf I}

\newcommand{\None}{t_{i}+1}
\newcommand{\Ntwo}{t_{i+1}}

\newcommand{\Pnew}{\bar P_{t+1}}

\newcommand{\Kold}{K_t}

\newcommand{\D}{D}

\renewcommand{\O}{\mathcal{O}}

\newcommand{\EPzero}{\mathcal{E}_{P_0}^{\rho}}

\usepackage{filecontents}

\usepackage{graphicx, xcolor}
\usepackage{multirow}
\usepackage{color, colortbl}
\definecolor{RoyalBlue}{rgb}{0.9,1,1}

\acrodef{lmis}[LMIs]{linear matrix inequalities}
\acrodef{adp}[ADP]{approximate dynamic programming}
\acrodef{cais}[CAIS]{constraint admissible invariant set}
\acrodef{sec}[SEC]{sufficient excitation condition}
\begin{document}
	
	\maketitle
	\begin{abstract}
	Enforcing state and input constraints during reinforcement learning (RL) in continuous state spaces is an open but crucial problem which remains a roadblock to using RL in safety-critical applications. This paper leverages invariant sets to update control policies within an approximate dynamic programming (ADP) framework that guarantees constraint satisfaction for all time and converges to the optimal policy (in a linear quadratic regulator sense) asymptotically. An algorithm for implementing the proposed constrained ADP approach in a data-driven manner is provided. The potential of this formalism is demonstrated via numerical examples.
	\end{abstract}
	
	\section{Introduction}
	Combining optimal control theory and reinforcement learning (RL) has yielded many excellent algorithms for generating control policies that imbue the closed-loop system with a desired level of performance in spite of unmodeled dynamics or modeling uncertainties~\cite{Lewis2009,lewis2012reinforcement}. Specifically,~\ac{adp} (sometimes also referred to as adaptive dynamic programming), a modern embodiment of RL~\cite{BertsekasBook,Kiumarsi2018} applied to continuous state and action spaces has gained traction for its ability to provide tractable solutions (in spite of the curse of dimensionality) to optimal control problems via function approximation and iterative updates of control policies and value functions~\cite{Hewer1971,kleinman1968iterative}.
	
	There are two main classes of~\ac{adp} algorithms: policy iteration and value iteration~\cite{sutton2011reinforcement}. A policy iteration algorithm for discrete-time linear systems was formulated in~\cite{bradtke1994adaptive} that leverages Q-functions proposed in~\cite{watkins1992q,werbos1989neural}, enabling control policy design without a complete system description. This methodology has been extended to continuous-time systems~\cite{vrabie2009adaptive}, $\mathcal H_2$ and $\mathcal H_\infty$ formulations~\cite{landelius1997reinforcement,al2007model}, tracking~\cite{kamalapurkar2015approximate}, output regulation~\cite{gao2016adaptive}, and  game-theoretic settings~\cite{zhang2017data,zhang2017discrete}. To reiterate, a particularly beneficial feature of this class of iterative methods is that control policies generated by policy  iteration converge to the optimal control policy with data obtained by  exciting the system dynamics, in spite of incomplete model knowledge~\cite{al2008discrete,Heydari2014}. 
	While optimality is important for certifying performance in a control system, often times the more critical concern is safety. A key aspect of safe control design is the ability of the system to respect both state and input constraints. To the best of our knowledge, this critical problem remains an open challenge in the context of ADP (and RL at large) in continuous state and action spaces.
	
	In this paper, we modify the classical policy iteration algorithm to incorporate safety through constraint satisfaction. The key idea is to compute control policies and associated constraint admissible invariant sets that ensure the system states and control inputs never violate design constraints. In the spirit of ADP, these policies and invariant sets are computed iteratively, and the sequence of policies are guaranteed to converge asymptotically to the optimal constraint-satisfying policy, provided that the system is sufficiently excited. The use of invariant sets to incorporate safety in learning/adaptive control algorithms via constraint handling has been done in model-based control design, such as model predictive control (MPC)~\cite{kothare1996robust,chakrabarty2017support,berntorp2017automated,Berkenkamp2016}, but its application to data-driven or model-free RL methods is relatively unexplored. A recent paper~\cite{li2018safe} is a noteworthy exception, although  our method is distinct from this work in that we do not compute a model of the system using the data obtained during operation; that is, our method is a direct data-driven approach, as defined in~\cite{piga2018direct}.
		
	The main \textbf{contributions} of this paper are: (i) we extend classical policy iteration in continuous state-action spaces to enforce state and input constraints; (ii) we provide a data-driven variant of this constrained policy iteration algorithm with unknown state matrix; and, (iii) we provide new sufficient conditions for safety (via constraint satisfaction), stability, and convergence of the policies generated by our proposed algorithm to the optimal constrained control policy.
	
	The rest of the paper is organized as follows. In Section~\ref{sec:notation} and~\ref{sec:motivation}, we present our notation and describe the problem statement in formal terms. We also discuss the standard (unconstrained) policy iteration algorithm. Our proposed alterations to the standard policy iteration algorithm for enabling state and input constraint satisfaction is discussed in Section~\ref{sec:constrained_ADP}. Theoretical performance certificates such as safety, stability, and algorithm convergence are provided in Section~\ref{sec:theorems}. Numerical examples including a 2D illustrative example and a 5D example are provided in Section~\ref{sec:ex} to illustrate the potential of this method. Conclusions are drawn in Section~\ref{sec:conc}.
	\section{Notation}\label{sec:notation}
	We denote by $\mathbb{R}$ the set of real numbers, $\mathbb R_+$ as the set of positive reals, and $\mathbb{N}$ as the set of natural numbers.
	For every $v\in\mathbb{R}^n$, we denote $\|v\|=\sqrt{v^\top v}$, where $v^\top$ is the transpose of $v$. The sup-norm is defined as $\|v\|_\infty \triangleq \sup_{t\in\mathbb{R}}\|v(t)\|$. We denote by $\underline\sigma(P)$ and $\overline\sigma(P)$ as the smallest and largest singular value of a square, symmetric matrix $P$, respectively. The symbol $\succ(\prec)$ indicates positive (negative) definiteness and $A\succ B$ implies $A-B\succ 0$ for $A,B$ of appropriate dimensions. Similarly, $\succeq (\preceq)$ implies positive (negative) semi-definiteness. The operator norm is denoted $\|P\|$ and is defined as the maximum singular value of $P$, $\mathrm{vec}(P)$ denotes the column-wise vectorization of $P$, and $\otimes$ denotes the Kronecker product. We parameterize an ellipsoid $\EPzero = \{x: x^\top P_0 x \le \rho\}$ using a scalar $\rho>0$ and a matrix $P_0\succ 0$.
	
	\section{Motivation}\label{sec:motivation}
	In this section, we begin by describing a general approximate dynamic programming formulation for solving the unconstrained discrete-time LQR problem.
	\subsection{Problem Statement}
	We consider discrete-time linear systems of the form
	\begin{equation}
	\label{eq:plant}
	x_{t+1} = Ax_t + Bu_t,
	\end{equation}
	where $t\in\mathbb R$ is the time index, $x\in \mathbb X\subset\mathbb R^n$ is the state of the system, $u\in\mathbb U\subset \mathbb R^m$ is the control input, and $x_{t_0}$ is a known initial state of the system. We assume the admissible state and input constraints sets $\mathbb X$ and $\mathbb U$ are polytopic, and therefore, can be represented as
	\begin{equation}\label{eq:constraints_polytopic}
	\mathcal X' = \left\{\begin{bmatrix}
	x \\ u
	\end{bmatrix}\in\mathbb{R}^{n+m}: c_i^\top x + d_i^\top u\le 1\right\},
	\end{equation}
	for $i=1,\ldots, r$, where $r$ is the total number of state and input constraints and $c_i\in\mathbb R^{n}$ and $d_i\in\mathbb R^{m}$.
	The sets $\mathbb X\subset \mathbb R^n$ and $\mathbb U\subset\mathbb R^m$ are known, compact, convex, and contain the origin in their interiors. Note that with any fixed control policy $K$, the constraint set described in~\eqref{eq:constraints_polytopic} is equivalent to the set
	\begin{equation}\label{eq:constraints_x}
	\mathbb X' = \left\{x\in\mathbb{R}^{n}: (c_i^\top + d_i^\top K) x \le 1\right\},
	\end{equation}
	for $i=1,\ldots, r$.
	\begin{remark}
	The inequalities~\eqref{eq:constraints_polytopic} define a polytopic admissible state and input constraint set. Note that $c_i=0$ implies that the $i$th constraint is an input constraint, and $d_i=0$ implies that it is a state constraint.
	\end{remark}
	The system matrix $A$ and input matrix $B$ have appropriate dimensions. 
	We make the following assumption on our knowledge of the system; these are standard assumptions in policy iteration methods.
	\begin{assumption}\label{asmp:A_unknown}
		The matrix $A$ is unknown, the matrix $B$ is known, and the pair $(A, B)$ is stabilizable. Furthermore, there exists a known control gain $K_0$ such that $u=K_0x$ is a stabilizing control policy for the system~\eqref{eq:plant}.
	\end{assumption}
	While the knowledge of the input matrix $B$ is not needed in approaches like Q-learning~\cite{al2007model}, it is fairly standard for policy improvement in policy iteration methods, even with function approximators~\cite{lewis2012reinforcement}. From a practical perspective, it is not uncommon for a designer to have knowledge of input channels and channel gains that represent the elements of the $B$ matrix.
	
	Our \textbf{objective} is to design an optimal control policy $\Kopt$ such that the state-feedback controller $u=\Kopt x$ stabilizes the partially known system~\eqref{eq:plant} while minimizing a cost functional
	\begin{equation}
	\label{eq:cost}
	V := \sum_{t=0}^\infty x_t^\top Q x_t + u_t^\top R u_t
	\end{equation}
	where $Q\succeq 0$ and $R\succ 0$ are user-defined symmetric matrices, with the pair $(A,Q^{1/2})$ being observable. The main contribution of this paper is to derive controller gains that stabilize the system~\eqref{eq:plant} while strictly enforcing state and input constraints.
	
	\subsection{Overview of optimal control for  discrete-time LQR}
	Let the value function be defined as
	\[
	V_{t}(x_t, u_t) := \sum_{k=t}^\infty x_k^\top Q x_k + u_k^\top R u_k.
	\]
	Clearly, $V_t$ satisfies the recurrence relation
	\begin{equation}\label{eq:recurrence}
	V_t(x_t, u_t) = x_t^\top Q x_t + u_t^\top R u_t + V_{t+1}(x_{t+1}, u_{t+1}).
	\end{equation}
	We know from optimal control theory that the optimization problem
	\begin{equation}\label{eq:optimal_value_fn}
	V_\infty(x_t) := \min_u V_t(x_t, u_t)
	\end{equation}
	is solved in order to obtain the optimal control action 
	\begin{equation}\label{eq:optimal_policy}
	u_\infty := \arg\min_u V_t(x_t, u_t)
	\end{equation}
	for each time instant $t\ge t_0$.
	For discrete-time linear time-invariant systems of the form~\eqref{eq:plant}, we know that the value function $V_t$ is quadratic in the state~\cite{lewis2012reinforcement}. Therefore, solving~\eqref{eq:optimal_value_fn} is equivalent to finding a symmetric matrix $\Popt \succ 0$ that satisfies the equation
	\begin{equation}
	A^\top \Popt A - \Popt + Q - A^\top \Popt B \left(R+ B^\top \Popt B\right)^{-1} B^\top \Popt A= 0.\label{eq:optimal_value_dt_lqr}
	\end{equation}
	Upon solving for $\Popt$, the optimal unconstrained discrete-time LQR policy generated by solving~\eqref{eq:optimal_policy} is given by
	\begin{equation}\label{eq:optimal_policy_dt_lqr}
	\Kopt = -(R+ B^\top \Popt B)^{-1}B^\top \Popt A.
	\end{equation}

	Since by assumption, the model $A$ is unknown, one cannot directly compute $\Popt$  from~\eqref{eq:optimal_value_dt_lqr} or $\Kopt$ from~\eqref{eq:optimal_policy_dt_lqr}. Instead, we resort to~\ac{adp}, an iterative method for `learning' the optimal control policy~\eqref{eq:optimal_policy_dt_lqr} by using on-line data without knowing a full model of the system~\eqref{eq:plant}. A popular embodiment of~\ac{adp} is \textit{policy iteration}, wherein an initial stabilizing control policy $K_0$ is iteratively improved using operational data, that is, without full model information. The sequence of control policies converges asymptotically to the optimal control policy $\Kopt$ defined in~\eqref{eq:optimal_policy_dt_lqr}. The key steps of policy iteration without constraints are described next.
	
	\subsection{Unconstrained policy iteration}
	Let $K_{t}$ be the $t$-th policy iterate, where $t \in \mathbb N$. Policy iteration has two key steps: policy evaluation and policy improvement. We begin by describing the steps in model-based policy iteration and subsequently demonstrate how to perform the same steps in a data-driven manner.
	\subsubsection{Model-based policy evaluation}
	In the policy evaluation step, the value function parameter $P_{t+1}\succ 0$ is estimated with the control gain $K_t$ using the relation
	\begin{equation}\label{eq:modelbased_policy_evaluation}
	(A+BK_t)^\top P_{t+1}(A+BK_t) - P_{t+1} + Q + K_t^\top R K_t = 0.
	\end{equation}
	Note that~\eqref{eq:modelbased_policy_evaluation} can be derived from~\eqref{eq:recurrence} when $V_t = x_t^\top P_t x_t$ and replacing $u_{t}=K_t x_t$ and $x_{t+1}=(A+BK_t)x_t$.
	\subsubsection{Model-based policy improvement}
	Upon updating the value function via~\eqref{eq:modelbased_policy_evaluation}, one needs to update the corresponding control policy. This is done by computing the new controller gain via
	\begin{equation}\label{eq:modelbased_policyimprov}
	K_{t+1} = -\left(R+B^\top P_{t+1} B\right)^{-1} B^\top P_{t+1}A.
	\end{equation}
	This equation is reminiscent of the optimal control policy equation~\eqref{eq:optimal_policy_dt_lqr}; in fact, the unique stationary point of the system of equations~\eqref{eq:modelbased_policy_evaluation}--\eqref{eq:modelbased_policyimprov} is at $P_t = P_\infty$ and $K_t = K_\infty$ as demonstrated in~\cite{kleinman1968iterative}.
	
	This model-based implementation can be performed in a data-driven manner, described next.
	\subsubsection{Data-driven policy evaluation}
	We assume that policy iteration is performed a discrete-time instances $t_i$ where 
	\begin{equation}\label{eq:T}
	\mathcal T=\{t_i\}_{i=1}^\infty
	\end{equation}
	denotes the set of all policy iteration times. The minimum number of data-points obtain between policy iterations $[t_i,t_{i+1}]$ is given by 
	\begin{equation}\label{eq:N}
	N = \inf_{i\in\mathbb N} \; \{t_{i+1} - t_i| t_i, t_{i+1}\in\mathcal T\},
	\end{equation}
	that is, $N$ denotes the minimum number of data points contained within any learning cycle. In a model based implementation, $\mathcal T=\mathbb N$.
	
	At each learning time instant $t_{i}\in\mathcal T$, one can rewrite~\eqref{eq:modelbased_policy_evaluation} as
	\begin{equation}\label{eq:modelfree_policy_eval}
	x_t^\top P^+ x_t = x_t^\top Q x_t + u^\top_t R u_t +x_{t+1}^\top P^+ x_{t+1},
	\end{equation}
	for every $t\in \{t_i+1, t_i+2,\ldots, t_{i+1}\}$, with $P^+$ representing the updated value function matrix.
	Assuming that the state and input data is available to us, and that $Q$ and $R$ are known, we can rewrite~\eqref{eq:modelfree_policy_eval} as 
	\begin{equation}\label{eq:datadrivenvecP}
	\Delta_{xx}\mathrm{vec}(P^+) = \begin{bmatrix}
	x_{t_i+1}^\top Q x_{t_i +1} + u_{t_i+1}^\top R u_{t_i+1} \\ 	x_{t_i+2}^\top Q x_{t_i +2} + u_{t_i+2}^\top R u_{t_i+2} \\ \vdots \\
	\\ x_{t_{i+1}}^\top Q x_{t_{i+1}} + u_{t_{i+1}}^\top R u_{t_{i+1}}
	\end{bmatrix},
	\end{equation}
	where
	\begin{equation}\label{eq:deltaxx}
	\Delta_{xx} = \begin{bmatrix}
	x_{t_i} \otimes x_{t_i} - x_{t_i+1} \otimes x_{t_i+1} \\ \vdots \\ x_{t_{i+1}}\otimes x_{t_{i+1}} - x_{t_{i+1}+1}\otimes x_{t_{i+1}+1}	\end{bmatrix}.
	\end{equation}

	Under well-known persistence of excitation conditions~\cite{lewis2012reinforcement}, one can solve~\eqref{eq:datadrivenvecP} as a (regularized) least squares problem subject to the constraint that $P^+\succ 0$ to obtain $P^+$ without knowing $A$ or $B$. For the time instants $t\in\mathcal T$ when the learning occurs, the new value function matrix $P_{t+1}$ is set to $P^+$ obtained by solving~\eqref{eq:datadrivenvecP}. For other time instants between learning time instants, that is $t\notin\mathcal T$, the value function matrix obtained in the previous learning cycle is utilized, that is, $P_{t+1} := P_{t}$.
	\subsubsection{Data-driven policy improvement}
	Since the control policy is restricted to be linear in this paper, finding an optimal policy is tantamount to finding the minimizer $K_{t+1}$ of the optimization problem
	\begin{equation}\label{eq:modelfree_policy_improvement}
	\min_{K} \sum_{t=t_i+1}^{t_{i+1}} \left(x_t^\top K^\top R K x_t + x_t^\top Q x_t + x_{t}^\top (A+BK)^\top P_{t+1} (A+BK)x_{t}\right),
	\end{equation}
	where $t_i, t_{i+1}\in\mathcal T$.
	This is a quadratic optimization problem in $K$ because $\{x_t\}$, $Q$, $R$, and $P_{t+1}$ are all known quantities in the window $\{t_i+1, t_i+2,\ldots, t_{i+1}\}$. 
	Note that $K_{t+1}$ can be updated recursively within each learning window $t_i\le t \le t_{i+1}$ using $P_{t+1}$ for these time instants. Since~\eqref{eq:modelfree_policy_improvement} is a quadratic problem, using Newton-type iterative solvers are expected to yield quick convergence; in this case, in one step.
	\section{Constrained ADP}\label{sec:constrained_ADP}
	In this section, we elucidate upon how to use invariant sets to generate new control policies that are both stabilizing and constraint satisfying. We also propose an algorithm for implementing a constrained ADP in a data-driven manner.
	
	We begin with the following definition.
		\begin{definition}[CAIS]
		\label{defn:invariant}
		A non-empty set $\mathcal E$ within the admissible state space $\mathbb X$ is a~\ac{cais} for the closed-loop system~\eqref{eq:plant} under a control law $u = Kx$ if, for every initial condition $x_{t_0}\in\mathcal E$, all subsequent states $x_t\in\mathcal E$ and inputs $Kx_t\in\mathbb U$ for all $t\ge t_0$.
	\end{definition}

According to Assumption 1, the ADP iteration is initialized with a stabilizing linear controller $K_0$.  This stabilizing controller renders a subset of the state-space invariant while satisfying state and input constraints. In particular, there exists an ellipsoidal region 
\[
\EPzero = \{x: x^\top P_0 x \le \rho\},
\] 
such that $\EPzero\subset \mathbb X$ and $K_0 \EPzero \subset \mathbb U$. We assume that the value function matrix $P_0$ defining the initial CAIS ellipsoid $\EPzero$ is known. This is encapsulated formally herein.

\begin{assumption}[Constrained ADP]
\label{asmp:constrainedADP}
There exists a symmetric positive definite matrix $P_0$ such that $\EPzero\subset \mathbb X$ is a CAIS for the closed-loop system (1) under the initial control policy $u=K_0 x$, and $K_0x \in\mathbb  U$ for all $x \in \EPzero$.
\end{assumption}

	\subsection{Model-based constrained policy iteration}
	\subsubsection{Model-based constrained policy evaluation}
	Let 
	\[
	\objfunc := (A+BK_t)^\top P(A+BK_t) - P + Q + K_t^\top R K_t.
	\]
	In order to implement constrained model-based policy evaluation (that is, obtain $P_{t+1}$ from $K_t$ and $P_t$), we need to solve the following semi-definite programming problem:
	\begin{subequations}
	\label{eq:modelbased_SDP}
	\begin{align}
	\label{eq:modelbased_SDP_1}
	&P_{t+1},\rho_{t+1} = \argmin_{P, \rho}  \|\objfunc\|\\
	\nonumber \text{subject to:}& \\
	\label{eq:modelbased_SDP_3}
	& (A+BK_t)^\top P(A+BK_t) - \lambda P \preceq 0\\
	\label{eq:modelbased_SDP_5}
	& x_t^\top P x_t \le \rho\\
	\label{eq:modelbased_SDP_4}
	& (c_k^\top+d_k^\top K_{t})^\top \rho\;(c_k^\top+d_k^\top K_{t})\preceq P\\
	\label{eq:modelbased_SDP_2}
	& \alpha_1\mathbf I \preceq P \preceq \alpha_2 \mathbf I\\
	\label{eq:modelbased_SDP_6}
	&\rho> 0
	\end{align}
	\end{subequations}
	for some $\alpha_1, \alpha_2 > 0$ and 
	\begin{equation}
	\label{eq:lambda_condition}
	\lambda<\left(\frac{\alpha_1}{\alpha_2}\right)^{2/N}.
	\end{equation}
	Here, $k\in\{1,\ldots, r\}$. Note that ensuring this problem is convex involves fixing the scalars $\alpha_1$ and $\alpha_2$, and pre-computing $\lambda$ using~\eqref{eq:lambda_condition}.
		
	The rationale behind~\eqref{eq:modelbased_SDP} can be explained as follows. Since~\eqref{eq:modelbased_SDP_2} ensures that $P\succ 0$, this constraint, along with the objective~\eqref{eq:modelbased_SDP_1}, is equivalent to~\eqref{eq:modelbased_policy_evaluation}, which is identical to the unconstrained policy evaluation step. Therefore, constraint satisfaction is made possible by equipping the constraints~\eqref{eq:modelbased_SDP_3}--\eqref{eq:modelbased_SDP_4} and~\eqref{eq:modelbased_SDP_6}.
	
	The inequality~\eqref{eq:modelbased_SDP_3} ensures that the value function is contractive, and therefore, non-increasing for every $t\ge t_0$. To see this, we multiply~\eqref{eq:modelbased_SDP_3} by $x^\top$ and $x$ from the left and right, respectively, which yields
	\[
	x_{t+1}^\top P x_{t+1} - x_t^\top P x_t  \le -(1-\lambda) \; x_t^\top P x_t  < 0,
	\]
	for any $t$, since $0<\lambda<1$. This is a key ingredient to ensure that the updated control policies will provide stability certificates for the closed-loop system. The two inequalities~\eqref{eq:modelbased_SDP_5} and~\eqref{eq:modelbased_SDP_4} enforce that the state and input constraints with the current policy are satisfied in spite of the value function update, given the current state $x_t$. The condition~\eqref{eq:modelbased_SDP_2} ensures that the value function matrix $P$ is positive definite, and the positive scalar $\rho$ allows the selection of sub- and super-level sets of the Lyapunov function. More details about how these conditions relate to theoretical properties of the proposed constrained ADP algorithm are provided in Section~\ref{sec:theorems}.
	\subsubsection{Model-based constrained policy improvement}
	Unlike unconstrained policy iteration, adding state and input constraints could result in nonlinear optimal control policies. In this paper, we restrict ourselves to design linear control policies of the form $u=Kx$, and hence, our optimal policy improvement step is analogous to the unconstrained case~\eqref{eq:modelbased_policyimprov}, that is,
	\begin{equation}\label{eq:modelbased_policyimprov_ideal}
	K^\star_{t+1} = -\left(R+B^\top P_{t+1} B\right)^{-1} B^\top P_{t+1}A.
	\end{equation}
	\begin{remark}
	In spite of parameterizing via linear control policies, our controller is actually nonlinear since $K_t$ depends on $P_t$ which depends on the states through~\eqref{eq:modelbased_SDP}.
	\end{remark}
	
	We adopt a backtracking strategy in order to update the current constrained policy $K_{t}$ to a new constrained policy $K_{t+1}$ that is as close as possible to the unconstrained policy $K^\star_{t+1}$ in~\eqref{eq:modelbased_policyimprov_ideal} that enforces state and input constraints. A simplified version of this backtracking strategy is outlined in Algorithm~\ref{alg:backtracking}. 
	
	\begin{algorithm}[!ht]
		\caption{Constrained Policy Improvement: Backtracking}
		\label{alg:backtracking}
		\begin{algorithmic}[1]
			\Require Desired policy $K^\star_{t+1}$ and current constrained policy $K_{t}$.
			\State $K_{t+1} \gets K^\star_{t+1}, \quad \alpha \gets 1$.
			\While {$(c_i^\top+d_i^\top K_{t+1})^\top \rho\;(c_i^\top+d_i^\top K_{t+1})\npreceq P_{t+1}$}
			
			\State {$\alpha \gets \beta \alpha$, where $0 < \beta < 1$.}
			\State {$K_{t+1} \gets K_{t} + \alpha \left(K^\star_{t+1} - K_{t}\right)$.}
			\EndWhile
		\end{algorithmic}
	\end{algorithm}
	\begin{remark}
	Note that the conditional statement in step~2 can be implemented efficiently based on a Cholesky factorization to check whether this particular symmetric matrix is positive definite.
	\end{remark}
	
	A particular benefit of our proposed method is that it enables both expansion, contraction, and rotation of the constraint admissible invariant sets. This is important in reference tracking for instance where a more aggressive controller is required when the state is near the boundary of the state constraints. This could also be useful for applying this approach to nonlinear systems where $(A,B)$ is a local linear approximation of the globally nonlinear dynamics. Our approach allows the ellipsoidal invariant sets to adapt its size and shape based on the local vector field. For example, suppose $\mathcal E_{\Popt}$ denote the CAIS that is associated with the constrained optimal control policy $\Kopt$ and optimal value function defined by $\Popt$. Also suppose that we have an initial admissible policy $K_0$ whose associated CAIS  $\EPzero$ is contained within $\mathcal E_{\Popt}$. Then our proposed method will generate a sequence of $\mathcal E_{P_t}$ such that these invariant sets will expand, contract, and rotate as necessary until the sequence of invariant sets $\{\mathcal E_{P_t}\}$ converges to the optimal $\mathcal E_{\Popt}$.

	\subsection{Data-driven constrained policy iteration}
	
	In order to obtain a data-driven implementation of the constrained ADP method, one needs to gather a sequence of state-input data points $\{\xk,\uk,\xkn \}$ and control policies $\{\Kk\}$ which will be used to update the value function matrix and control policies at the learning time instants defined by $\mathcal T$ in~\eqref{eq:T}. Given the discrete-time system dynamics in~\eqref{eq:plant}, the relation between these data points is given by
	\begin{equation}
	\xkn = A \xk + B \uk = A \xk + B \left(\Kk \xk +\unoise\right),  \label{eq:plant_meas}
	\end{equation}
	where $\unoise$ represents a known exploration noise signal that ensures the system~\eqref{eq:plant_meas} is persistently excited; see~\cite{lewis2012reinforcement}. To arrive at a more compact notation, let us define $$\xknt := \xkn - B \unoise \;\;\text{and} \;\;\ukt := \Kk\, \xk$$ such that 
	\begin{equation}
	\label{eq:tilde_bar_relationship}
	\xknt = A \xk + B \ukt = (A+B\Kk)\xk.
	\end{equation} 
	
	\subsubsection{Data-driven constrained policy evaluation}
	Consider the $i$-th learning cycle, occuring at the time instant $t_i\in\mathcal T$. Let
	\[
	\ztp := \xknt^\top P \xknt - \xk^\top P \xk + \xk Q \xk + \ukt^\top R\, \ukt. 
	\]
	The data-driven analogue of the constrained policy evaluation step discussed in the previous section is given by the following semi-definite program~(SDP) with $\alpha_1$ and $\alpha_2$ fixed:
	\begin{subequations} \label{SDP}
		\begin{align}
		&\Pnew, \rho_{t+1} :=
		\arg\min_{\rho, P} \;\;\frac{1}{2}\, \sum_{t=0}^{\None-1} \left(\ztp\right)^2 - \lambda_\rho \rho \label{SDP:obj}\\
		\text{subject to:}& \nonumber\\
		& \xknt^\top P \xknt - \lambda \xk^\top P \xk \le 0 \label{SDP:Lyapunov}\\
		& x_{t_{i+1}}^\top P x_{t_{i+1}} \le \rho
		\label{SDP:x_t_P}\\
		& 
		(c_k^\top+d_k^\top \Kk)^\top \rho\;(c_k^\top+d_k^\top \Kk)\preceq P
		\label{SDP:mixedConstr}\\
		& \alpha_1\mathbf I \preceq P \preceq \alpha_2 \mathbf I\\
		& \rho > 0,
		\end{align}
	\end{subequations}
	for $t \in \{t_i+1,t_i+2,\ldots,t_{i+1}\}$ and $k \in \{1,\ldots,r\}$.
	Note that the final four inequalities in~\eqref{SDP} are exactly the set of inequalities presented in~\eqref{eq:modelbased_SDP} with the model information replaced by state and input data. Also, replacing $\xknt$ in~\eqref{SDP:Lyapunov} with $(A+B\Kk)\xk$ using  equality~\eqref{eq:tilde_bar_relationship} shows that it is equivalent to the inequality~\eqref{eq:modelbased_SDP_3}.
	
	\subsubsection{Data-driven constrained policy improvement}
	Once a value function is found whose sub-level set is constraint admissible, the corresponding policy $K_{t+1}$ is to be computed. If $A$ and $B$ are known, this step would be easy: indeed, one could utilize Eq.~\eqref{eq:modelbased_policyimprov_ideal} to this end.
	However, since only $B$ is known (by assumption), we resort to a data-driven iterative update methodology for generating the new policy. 
	
	Given the current policy $\Kold$, we gather another batch of measurements $\{ \xk,\uk,\Kk,\xkn \}_{t=t_i+1,\ldots,t_{i+1}}$ where a new policy $\bar K_t$ is the optimizer of the least squares problem
 	\begin{align}
		\label{LLS}
		\underset{K}{\text{min}}\;\; & \frac{1}{2} \sum_{t=t_i+1}^{t_{i+1}}  \xk^\top \left(K^\top R K + (A+BK)^\top \Pnew (A+BK)\right)\xk.
	\end{align}
	The problem~\eqref{LLS} can be solved in a data-driven manner efficiently using a real-time recursive least squares (RLS) implementation~\cite{Ljung1999} 
	\begin{subequations} \label{RLS}
		\begin{alignat}{5}
		H_{t+1} &= H_{t} + \xk \xk^\top \otimes (R+B^\top \Pnew B), \label{hess_up} \\
		g_{t+1} &= \xk \otimes (R \Kk \xk + B^\top \Pnew \xknt ), \label{grad_up} \\
		\mathrm{vec}(\bar{K}_{t+1}) &= \mathrm{vec}(\bar{K}_{t}) - \beta_{t}\, H_{t+1}^{-1} \, g_{t+1}, \label{pol_up}
		\end{alignat}
	\end{subequations}
	for $t=\None,\ldots,\Ntwo-1$. Note that~\eqref{RLS} is solved without knowledge of $A$ using the updates. Also,  the starting value for the Hessian matrix is chosen as the identity matrix $\rho\, \eye$ and $\rho > 0$ to ensure non-singularity. The step size $\beta_{t}$ is typically equal to one, even though a smaller step $\beta_{t} \le 1$ can be chosen, e.g., based on the backtracking procedure in Algorithm~\ref{alg:backtracking} in order to impose the affine state and input constraints in~\eqref{SDP:mixedConstr} for each updated control policy $\bar{K}_{t+1}$. The Hessian matrix in~\eqref{RLS} can be reset to $H = q\, \eye\succ 0$ whenever a new value function is obtained from solving the SDP in~\eqref{SDP}. Note that~\eqref{hess_up} corresponds to a rank-$m$ matrix update, where $m$ denotes the number of control inputs. Therefore, its matrix inverse $H_{t+1}^{-1}$ can be updated efficiently using the Sherman-Morrison formula, for example, in the form of $m$ rank-one updates.
	

	\subsubsection{Algorithm Implementation: Pseudocode}
	
	Algorithm~\ref{alg:CADP} provides a detailed description of our proposed approach for data-driven constrained adaptive dynamic programming for linear systems. The general procedure corresponds to the sequence of high-level steps:
	\begin{enumerate}[(i)]
		\item We require an initial stabilizing policy $\bar{K}_0$ and a corresponding constraint admissible invariant set~(CAIS) $\EPzero$; see Assumptions 1 and~\ref{asmp:constrainedADP}.
		\item Obtain a sequence of at least $\None$ data points $\{ \xk,\uk,\Kk,\xkn \}$ while the system is persistently excited and compute a new ellipsoidal set defined by the matrix $\Pnew$ by solving the least squares SDP in~\eqref{SDP}.
		\item At each time step, perform the policy improvement step to compute $\bar{K}_{t+1}$ based on the real-time recursive least squares method as described in~\eqref{RLS}, in combination with the backtracking procedure of Algorithm~\ref{alg:backtracking} to enforce state and input constraints.
		\item If the policy improvement has converged based on the condition $\|g_t\|\le \epsilon$, return to step~(ii).
	\end{enumerate}
	
	\begin{algorithm}[!ht]
		\caption{Data-driven constrained~\ac{adp}}
		\label{alg:CADP}
		\begin{algorithmic}[1]
			\Require Initial policy $\bar{K}_0$ and~\ac{cais} $\EPzero$ (see Assumption~\ref{asmp:constrainedADP}), initial state value $\bar{x}_0$ and $\epsilon > 0$.
			\State $\D \gets \{\}$.
			\For {$t = 0, 1, \ldots$}
			
			\State {Apply control input $\uk = \unoise + \Kk \xk$.} \Comment{to system} \vspace{1mm}
			\State Obtain new state estimate $\xkn$. \Comment{from system} \vspace{1mm}
			
			\If {$\Vert g_{t} \Vert \le \epsilon$}: \Comment{convergence check}
			\State {$\D \gets \{ \D, t \}$}. \Comment{add data time stamp}
			\EndIf
			
			\Statex \texttt{\hspace{1.5em}/* Policy evaluation step~(SDP) */}
			\If {PE condition holds with data $\forall\; t\in D$}:
			\State {Compute $\rho^+, P^+$ by solving SDP~\eqref{SDP} based on}
			\Statex {\hspace{3em}stored buffer of data points $\{ \xk,\uk,\Kk,\xkn \}_{t \in \D}$.}
			\State $\bar{P}_{t+1}, \rho_{t+1} \gets P^+, \rho^+$. \Comment{define new ellipsoid}
			\State Reset buffer $\D \gets \{\}$.			
			\Else
			\State $\bar{P}_{t+1}, \rho_{t+1} \gets \bar{P}_{t}, \rho_t$.
			\EndIf
			
			\Statex \texttt{\hspace{1.5em}/* Policy improvement step~(RLS) */}
			\State {Compute new policy $\bar{K}_{t+1}$ as in~\eqref{RLS}, using Alg.~\ref{alg:backtracking},} 
			\Statex {\hspace{1.5em}given new measurements $(\xk,\uk,\Kk,\xkn)$ and $\bar{P}_{t+1}$.}\vspace{1mm}
			
			\EndFor
		\end{algorithmic}
	\end{algorithm}

	\subsection{Remarks on computational complexity}
	
	Semidefinite programs~(SDP) of the form~\eqref{SDP} are convex optimization problems that can be solved in polynomial time, for example, using interior point methods (IPMs). However, in general, standard implementations of IPMs for solving SDPs have a computational complexity $\O(n^6)$ when solving for $n \times n$ matrix variables and a memory complexity of $\O(n^4)$~\cite{Todd2001,Vandenberghe1996}. Instead, the per iteration complexity and memory requirements for first order optimization algorithms such as, e.g., the alternating direction method of multipliers~(ADMM) can be much smaller, even though they typically require more iterations in practice~\cite{Wen2010,Zheng2017}. Note that, instead, the policy improvement steps are computationally cheap because both the low-rank update techniques for the Hessian matrix~\eqref{hess_up} and the matrix-vector multiplication in~\eqref{pol_up} can be performed with a complexity $\O(n^2m^2)$ that scales quadratically with the dimensions of the policy matrix $K$.
		
	The policy evaluation step, based on the SDP solution in~\eqref{SDP}, could be computed also using a recursive least squares type implementation or in a receding horizon or sliding window manner. However, given the computational complexity of treating linear matrix inequalities in the SDP formulation, a batch-type approach as in Algorithm~\ref{alg:CADP} would typically be preferred for real-time feasible control applications under strict timing requirements. Additionally, it is important to note that the SDP solution in Algorithm~\ref{alg:CADP} is not necessarily required to be real-time feasible, unlike the recursive least squares based policy improvement step in~\eqref{RLS}, which is computationally cheap. 
	
	\section{Constraint Satisfaction,  Stability, and Algorithm Convergence}\label{sec:theorems}
	We present theoretical guarantees for our proposed constrained policy iteration. For the data-driven case, we adhere to the standard assumption that the system is persistently excited. The following theorem demonstrates constraint enforcement and stability guarantees of the closed-loop system.
	\begin{theorem}
		Suppose Assumptions~\ref{asmp:A_unknown} and~\ref{asmp:constrainedADP} hold. Then the system~\eqref{eq:plant} in closed-loop with the time-varying controller $u_t = K_t x_t$ has the following properties:
		\begin{enumerate}[(i)]
			\item The constraints $x_t \in \mathbb{X}$ and $u_t \in \mathbb{U}$ are satisfied for all $t \in \mathbb{N}$.
			\item The closed-loop system is asymptotically stable.
		\end{enumerate}
	\end{theorem}
	
	\begin{proof}
	(i) Consider the ellipsoid $\mathcal E_{P_{t+1}}^\rho$.
	The inequality~\eqref{eq:modelbased_SDP_4} yields
	\begin{align*}
	\frac{1}{\rho} P_{t+1} \succeq (c_i^\top + d_i^\top K)^\top (c_i^\top + d_i^\top K).
	\end{align*}
	For $x \in \mathcal E_{P_{t+1}}^\rho$ we have $x^\top P x \leq \rho$. Thus, $$
	x^\top (c_i^\top + d_i^\top K)^\top (c_i^\top + d_i^\top K) x \leq 1
	$$ which implies $x \in \mathbb{X}'$ since $(c_i^\top + d_i^\top K) x \leq 1$. Thus, $\mathcal E_{P_{t+1}}^\rho \subseteq \mathbb{X}'$, which implies that state and input constraints are satisfied for all states inside the ellipsoid $\mathcal E_{P_{t+1}}^\rho$. 
	
	Note that the state $x_t$ is contained in this ellipsoid while the $t$-th controller $u_t = K_t x_t$ is active. This can be inferred from~\eqref{eq:modelbased_SDP_3} and~\eqref{eq:modelbased_SDP_5}, which (respectively) imply that the ellipsoid $\mathcal E_{P_{t+1}}^\rho$ is positive-invariant and that the initial state is contained in the ellipsoid when the controller is first engaged.
			
	(ii) Since the closed-loop system is a switched system, we will use the concept of multiple Lyapunov functions to prove stability. Consider the set of Lyapunov functions $$V_P(x) = x^\top P x$$ for all $P$ that satisfies~\eqref{eq:modelbased_SDP}. We will show that the $t$-th controller $u_t = K_t x_t$ decreases all of these Lyapunov functions\footnote{The Lyapunov functions do not necessarily decrease monotonically for all $t$ as long as the function values are decreasing at each learning instant $t_i$.} over the time period $t \in \{t_i+1, t_{i}+2,\ldots, t_{i+1}\}$ for the $i$-th learning cycle in which it was engaged (recall $t_i, t_{i+1}\in\mathcal T$ defined in~\eqref{eq:T}). Note
	\begin{align*}
	V_P\left(x_{{t_{i+1}}}\right) &\leq \frac{\alpha_2}{\alpha_1} V_{P_{t+1}}\left(x_{{t_{i+1}}}\right) \\
	&\leq \frac{\alpha_2}{\alpha_1} \lambda^N V_{P_{t+1}}\left(x_{{t_i}}\right) \\
	&< \frac{\alpha_1}{\alpha_2} V_{P_{t+1}}\left(x_{{t_i} }\right) \\
	&\leq V_{P}\left(x_{{t_i}}\right),
	\end{align*}
	where the first and last inequalities are a consequence of~\eqref{eq:modelbased_SDP_2}. The second inequality is a consequence of~\eqref{eq:N} and~\eqref{eq:modelbased_SDP_3}, along with Algorithm 1 which uses a convex combination of $K_t$ and $K_{t+1}^\star$, both of which are guaranteed to contract the value function by $\lambda$. The third inequality is a consequence of the condition~\eqref{eq:lambda_condition}. Since $V_P\left(x_{{t_{i+1}}}\right)<V_{P}\left(x_{{t_i}}\right)$, each Lyapunov function $V_P(x)$ converges to zero. As~\eqref{eq:modelbased_SDP_2} ensures that these Lyapunov functions are positive-definite, we get $x_t\to 0$ as $t\to\infty$, which concludes the proof.
	\end{proof}
	
	Previous stability results for approximate dynamic programming rely on the tacit assumption that the learning converges after a finite number of batch iterations (typically one). In other words, the adaptive controller only works because it stops adapting. In contrast, for constraint satisfaction, the controller may need to continually adapt since the set of active constraints will change as the state evolves. This necessitates the development of a more involved set of conditions to ensure that feedback control loop and the learning loop do not destabilize each other. 
	
	\begin{theorem}
	Suppose Assumptions~\ref{asmp:A_unknown} and~\ref{asmp:constrainedADP} hold. Let $\alpha_1 \leq \underline{\sigma}(\Popt)$, $\alpha_2 \geq \underline{\sigma}(\Popt)$, and $$\lambda \geq \overline{\sigma}\left( I - \Popt^{-1/2} (Q+K_\infty^\top R K_\infty) \Popt^{-1/2}\right).$$ Under the iteration~\eqref{eq:modelbased_SDP} and~\eqref{eq:modelbased_policyimprov_ideal}, the value $P_t$ and policy $K_t$ converge to the LQR cost-to-go $\Popt$ and controller gain $\Kopt$. That is,
	\begin{equation}\label{eq:thm2_convergence}
	\lim_{t\to\infty} P_t = \Popt \quad \text{and}\quad \lim_{t\to\infty} K_t = K_\infty.
	\end{equation}
	\end{theorem}	
	\begin{proof}
	If feasible, the LQR cost-to-go $\Popt$ with controller gain $\Kopt$ will be the optimal solution of~\eqref{eq:modelbased_SDP}. By the assumptions on $\alpha_1,\alpha_2$, the LQR value $\Popt$ satisfies~\eqref{eq:modelbased_SDP_2}. Note that
	\begin{align*}
	&\overline{\sigma}\big( I - \Popt^{-1/2} (Q+\Kopt^\top R \Kopt) \Popt^{-1/2}\big) \\
	&\qquad= \sup_z \frac{ z^\top (I - \Popt^{-1/2} (Q+\Kopt^\top R \Kopt) \Popt^{-1/2}) z } {z^\top z}\\
	&\qquad= \sup_x \frac{ x^\top (\Popt - (Q+\Kopt^\top R \Kopt)) x } {x^\top \Popt x}
	\end{align*}	
	where $z = \Popt^{-1/2} x$. Thus, the LQR satisfies
	\begin{align*}
	&(A+B\Kopt)^\top P_{\infty} (A+B\Kopt) \\
	& \qquad = \Popt - (Q + \Kopt^{\top} R \Kopt) \\
	&\qquad \leq \overline{\sigma}\big( I - \Popt^{-1/2} (Q+\Kopt^\top R \Kopt) \Popt^{-1/2}\big) \Popt \\
	&\qquad \leq \lambda \Popt.
	\end{align*}
	Thus, the LQR value $\Popt$ and policy $\Kopt$ satisfy~\eqref{eq:modelbased_SDP_3}.
		
	Finally, we show that~\eqref{eq:modelbased_SDP_4} and~\eqref{eq:modelbased_SDP_5} are satisfied at some finite time $T\ge t_0$. Since, the state $\mathbb{X}$ and input $\mathbb{U}$ constraints contain the origin in their interiors, there exists $\rho > 0$ such that the ellipsoidal region $\mathcal E_{P_{\infty}}^\rho$ of the LQR cost-to-go $x^\top \Popt x$ satisfies~\eqref{eq:modelbased_SDP_4}. 
		
	Furthermore, since the closed-loop system is asymptotically stable and $\mathcal E_{P_{t}}^\rho$ is a CAIS for every $t\ge t_0$, there exists a finite time $T \ge t_0$ such that $x_t \in \mathcal E_{P_{t}}^\rho$ for all $t\ge T$. Thus,~\eqref{eq:modelbased_SDP_5} will be satisfied by the LQR controller after time $T$. As for all $t\ge T$, the state and input constraints are automatically satisfied (and are therefore, inactive), one can use the same arguments as in classical model-based policy iteration~\cite{Hewer1971} to conclude the proof.
	\end{proof}
	\begin{remark}
	In order to solve SDP~\eqref{SDP} using least squares methods, one need to collect a sequence of states such that the matrix $\Delta_{\bar x \bar x}$ (obtained by replacing $x$ in~\eqref{eq:deltaxx} with $\bar x$) has full column rank. This full rank condition is like the condition of persistence excitation (PE) in adaptive control theory. In order to satisfy this full rank condition, we add an exploration noise $\nu_t$ into the input to excite the system as in~\cite{lewis2012reinforcement,bradtke1994adaptive,al2007model}. As the exploration noise $\nu_t$ goes to zero, the solution to~\eqref{SDP} will converge to the solution to model-based constrained policy evaluation problem~\eqref{eq:modelbased_SDP}.  
	\end{remark}%
	\section{Numerical Example}\label{sec:ex}
	\subsection{Linear system with two states, one control input}
	We randomly generate controllable systems of the form~\eqref{eq:plant} to test the proposed algorithm. A particular realization of these randomly generated systems,
	$A = \left[\begin{smallmatrix} 1.1387 & 0.0491\\
	-0.8680 &    0.9679\end{smallmatrix}\right], \quad B = \big[\begin{smallmatrix}
	-0.5507 \\
	0.0758
	\end{smallmatrix}\big]$
	is investigated to illustrate constraint satisfaction and stability of the algorithm. Of course, $A$ is unknown (and unstable), $B$ is known, and it is verified that $(A,B)$ is a controllable pair. The admissible state space is given by $\mathbb X=\{x\in\mathbb R^2: \|x\|_\infty \le 1\}$, and the operational cost is parameterized by $Q=I_2$ and $R=0.5$.
	\begin{figure}[!ht]
		\centering
		\includegraphics[width=0.9\columnwidth]{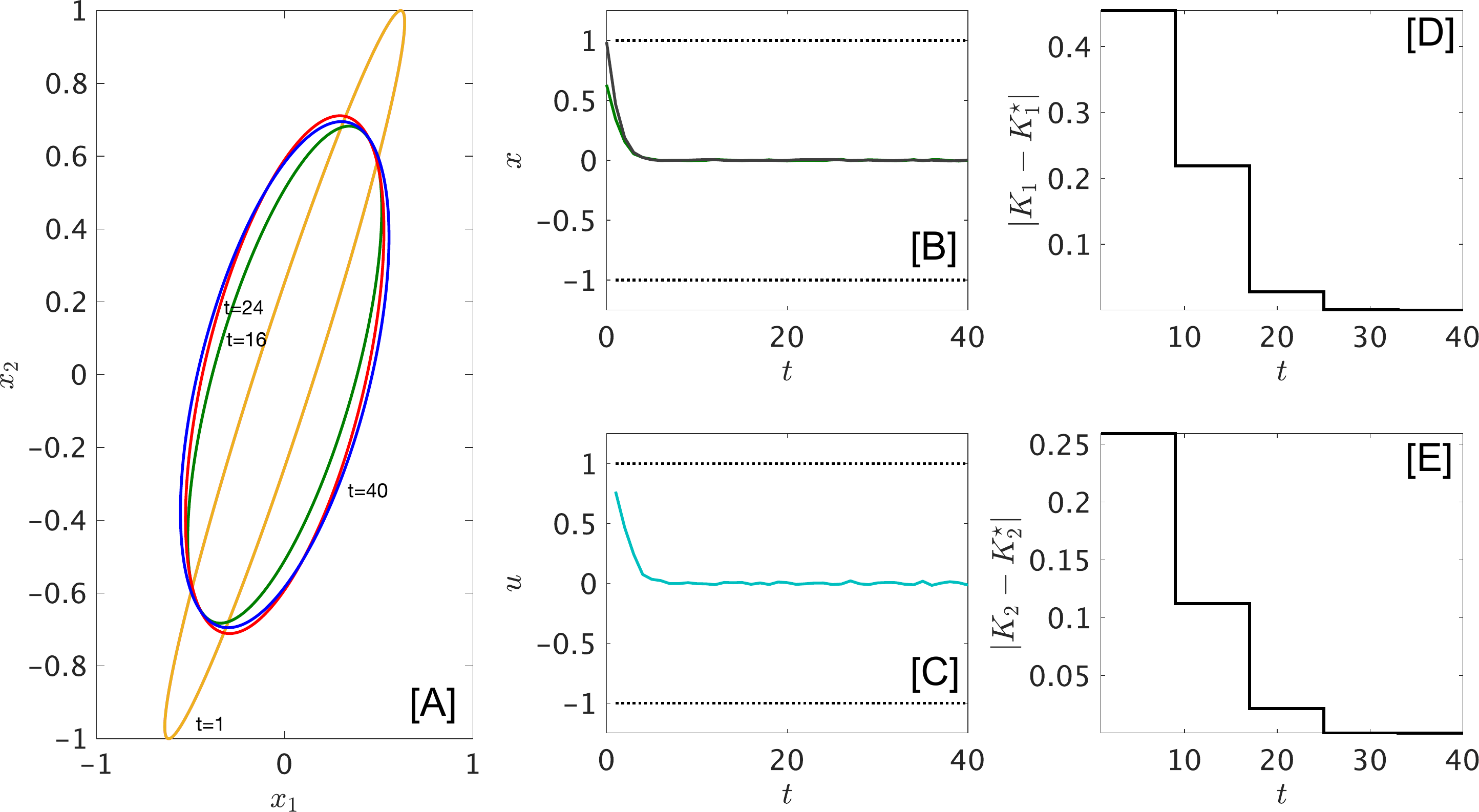}%
		\caption{Results of constrained ADP for $2$-state dynamic system: [A] Sequence of invariant sets learned on-line. Each set is labeled with the time iteration $t$ when it was learned. [B] State evolution ($x_1$: blue, $x_2$: red) with constraints (black, dashed). [C] Control input (blue) evolution with constraints (black, dashed). [D, E] Convergence of learned LQR policy to the true LQR policy.}
		\label{fig:fig1}
	\end{figure}%
	For learning, the window length is fixed at $N=8$ samples ($\mathcal T = \{8, 16, 24,\ldots\}$), and the regularization parameter for policy updating is given by $\rho_K=10^{-4}$. Persistence excitation is ensured by generating uniformly distributed noise bounded within $[-0.02, 0.02]$. An initial policy is generated that satisfies state constraints using the randomly chosen cost matrices that are distinct from $Q$ and $R$, and an initial condition is generated randomly on the boundary of the initial domain of attraction. Therefore, the initial state is ensured to be within $\mathbb X$ but sufficiently far from the origin to require non-trivial control for stabilization.

	The results of the constrained policy iteration algorithm are illustrated in Fig.~\ref{fig:fig1}. In Fig.~\ref{fig:fig1}[A], a sequence of ellipsoids generated by our proposed algorithm is presented. Note that the ellipsoids generated in subsequent learning cycles after the first (the orange elongated ellipsoid) are not mere sub- or super-level sets of the initial ellipsoid; instead, the policy iterator allows for contractions and expansions on both $x_1$ and $x_2$ axes until the true policy is learned. As evident from subplots [B] and [C], state constraints are not violated throughout the learning procedure. The subplots [D, E] demonstrate the convergence of a sub-optimal initial control policy at $t=0$ to the true and optimal LQR policy at around $t=24$, after three learning cycles.
	
	\subsection{Higher-dimensional linear system}
		\begin{figure}[!ht]
		\centering
		\includegraphics[width=.7\columnwidth]{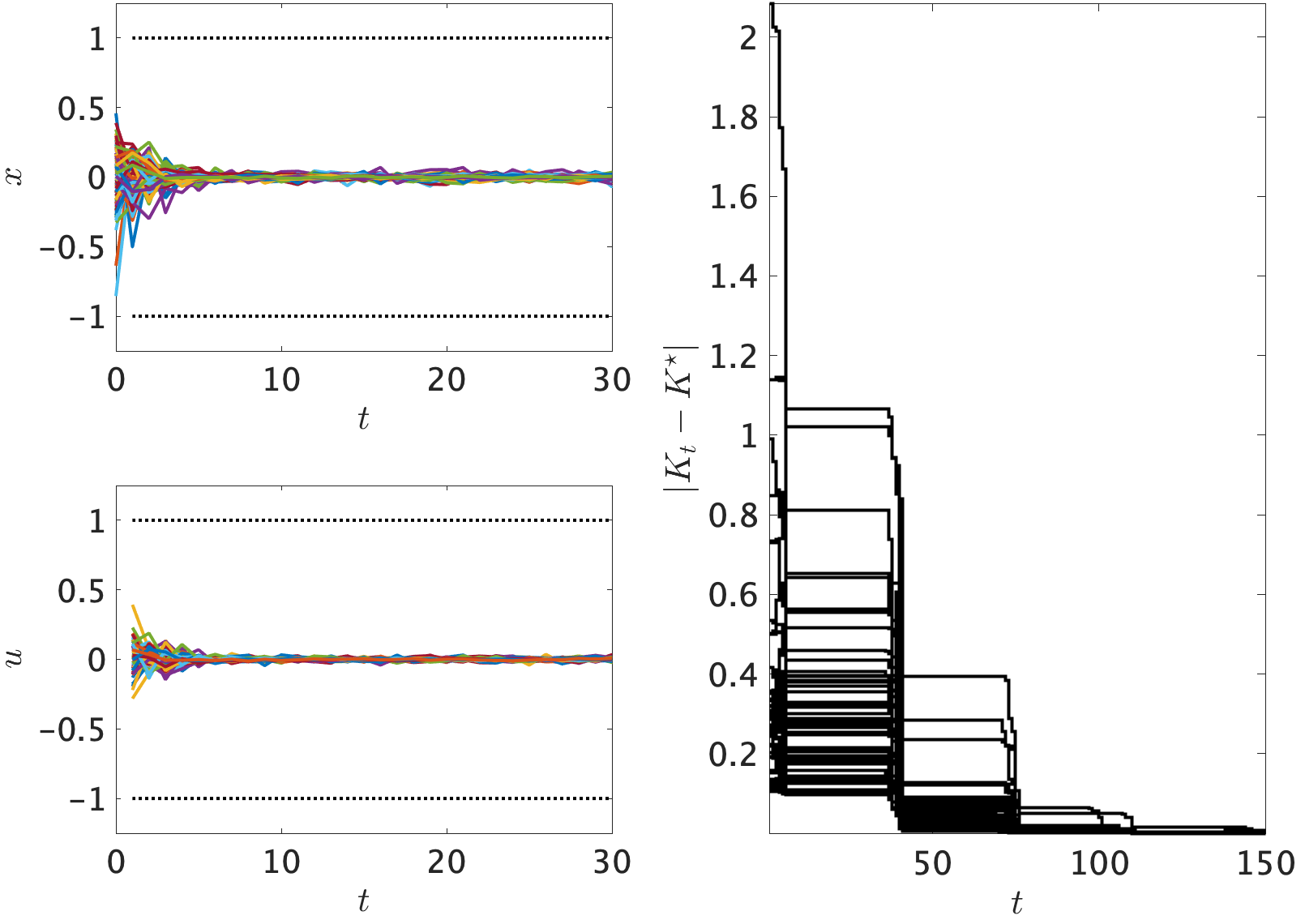}
		\caption{Results of constrained ADP for $5$-state, $2$-input dynamic system: [A] State evolution with constraints (black, dashed). [B] Control input evolution with constraints (black, dashed). [C] Convergence of learned LQR policy to the true LQR policy, using $2$-norm of the error.}
		\label{fig:fig2}
	\end{figure}%
	In addition, let us illustrate the performance for a dynamic system with $5$ states and $2$ control inputs. For this purpose, we have randomly generated $50$ unique dynamic systems~\eqref{eq:plant} that are constructed to be (slightly) unstable but controllable, such that an initial stabilizing policy can be obtained. Similar to before, the system matrix $A$ is unknown but the input matrix $B$ is known. The admissible state and input space, respectively, is given by $\mathbb X=\{x\in\mathbb R^5: \|x\|_\infty \le 1\}$ and $\mathbb U=\{u\in\mathbb R^2: \|u\|_\infty \le 1\}$, and the cost matrices read as $Q=I_5$ and $R=0.5 I_2$. Figure~\ref{fig:fig2} presents the closed-loop state and input trajectories when applying the proposed constrained ADP implementation~(see Algorithm~\ref{alg:CADP}) to each of these generated test problems. Note that the learning window has been chosen to be equal to $N=30$ samples; indicating that $\mathcal T = \{30, 60,\ldots\}$.
	
	From Figure~\ref{fig:fig2}, it can be observed that the $5$-dimensional dynamic system is stabilized (the small perturbations in subplot [A] and [B] are due to the exploratory noise) in all of the generated test cases and both the state and input constraints are respected at all time. In addition, in most of the cases, the optimal policy is obtained relatively quickly in an amount of time that corresponds to $2$ learning window lengths, i.e., $60$ time steps in Figure~\ref{fig:fig2}. The policy error, computed using the matrix $2$-norm of the difference between the current and optimal policy, generally decreases over time for all cases, under the necessary condition for persistence of excitation.	
	\section{Concluding Remarks}\label{sec:conc}
	In this paper, we provide a methodology for implementing constraint satisfying policy iteration for continuous-time, continuous-state systems via invariant sets. Benefits of our approach include computational tractability, and safety guarantees through constraint satisfaction. In future work, we will extend this framework to nonlinear systems.
	
	\renewcommand{\baselinestretch}{0.99}
	\bibliographystyle{IEEEtran}
	\bibliography{IEEEabrv,\jobname}
\end{document}